\documentclass[lettersize,journal]{IEEEtran}
\usepackage[english]{babel}
\usepackage{amsmath,amsfonts}
\usepackage{amsthm} 
\usepackage{amssymb} 
\usepackage{mathtools} 
\usepackage{algorithmic} 
\usepackage{algorithm}
\usepackage{array}
\usepackage[caption=false,font=normalsize,labelfont=sf,textfont=sf]{subfig}
\usepackage{textcomp}
\usepackage{stfloats}
\usepackage{url}
\usepackage{verbatim}
\usepackage{graphicx}
\usepackage{cite}
\usepackage{lipsum}
\usepackage{hyperref}
\usepackage{soul}
\usepackage{xcolor}
\usepackage{balance}

\usepackage{epsfig} 
\usepackage{multirow} 
\usepackage{multicol} 
\usepackage{makecell} 
\usepackage[short]{optidef} 

\newtheorem{Lemma}{Lemma}

  {\proof}{\proofend}
\newtheorem{proposition}{Proposition}

\newcommand{\qa}{{\bf a}}

\newcommand{\qe}{{\bf e}}

\newcommand{\qg}{{ \textbf{g} }}
\newcommand{\qh}{{ \textbf{h} }}

\newcommand{\qu}{{\bf u}}

\newcommand{\qw}{{\bf w}}
\newcommand{\qx}{{\bf x}}
\newcommand{\qy}{{ \textbf{y} }}
\newcommand{\qz}{{ \bf z }}

\newcommand{\qB}{{\bf B}}

\newcommand{\qF}{{\bf F}}
\newcommand{\qG}{{ \textbf{G} }}
\newcommand{\qH}{{ \textbf{H} }}
\newcommand{\qI}{{ \textbf{I} }}

\newcommand{\qM}{{\bf M}}
\newcommand{\qN}{{\bf N}}

\newcommand{\qR}{{\bf R}}

\newcommand{\qU}{{\bf U}}
\newcommand{\qV}{{\bf V}}

\newcommand{\qX}{{\bf X}}

\newcommand{\qZ}{{ \textbf{Z} }}

\DeclareMathOperator*{\argmax}{arg\,max}

\newcommand{\UE}{\mathtt{I}}
\newcommand{\sn}{\mathtt{E}}

\DeclareMathOperator{\ETAI}{\boldsymbol{\eta}^{\mathtt{I}}}
\DeclareMathOperator{\ETAE}{\boldsymbol{\eta}^{\mathtt{E}}}
\DeclareMathOperator{\MM}{\mathcal{M}}

\DeclareMathOperator{\K}{\mathcal{K}}

\DeclareMathOperator{\J}{\mathcal{J}}

\DeclareMathOperator{\C}{\mathbb{C}}
\DeclareMathOperator{\CN}{\mathcal{CN}}

\newcommand{\PZF}{\mathsf{PZF}}
\newcommand{\PMRT}{\mathsf{PMRT}}

\newcommand{\wimk}{\qw_{\mathrm{I},mk}}

\newcommand{\wemj}{\qw_{\mathrm{E},mj}}

\newcommand{\wimkp}{\qw_{\mathrm{I},mk'}}

\newcommand{\wemjp}{\qw_{\mathrm{E},mj'}}
\newcommand{\Ghms}{\hat{\qG}_m^{\sn}}
\newcommand{\Ghmu}{\hat{\qG}_m^{\UE}}
\newcommand{\Snn}{\sigma_n^2}
\newcommand{\Ex}{\mathbb{E}}

\newcommand{\yej}{y_{\mathtt{E},j}}

\newcommand{\yik}{y_{\mathtt{I},k}}

\newcommand{\gmkiu}{\qg_{mk}^{\UE}}
\newcommand{\gmjue}{\qg_{mj}^{\sn}}

\newcommand{\hgmjue}{\hat{\qg}_{mj}^{\sn}}
\newcommand{\hmlue}{\qh_{mj}^{\sn}}
\newcommand{\hmlueo}{\qh_{m1}^{\sn}}
\newcommand{\hmlueJ}{\qh_{mJ}^{\sn}}
\newcommand{\Hmlue}{\qH_{m}}

\newcommand{\hmueo}{\qh_{m1}}
\newcommand{\hmueL}{\qh_{mL}}
\newcommand{\hmuell}{\qh_{ml}}

\newcommand{\trace}{\mathrm{tr}}

\newcommand{\hgmkue}{\hat{\qg}_{mk}^{\UE}}

\newcommand{\ghmons}{\hat{\qg}_{m1}^{\sn}}

\newcommand{\ghmJue}{\hat{\qg}_{mJ}^{\sn}}
\newcommand{\ghmonue}{\hat{\qg}_{m1}^{\UE}}
\newcommand{\ghmKue}{\hat{\qg}_{mK}^{\UE}}
\newcommand{\ghmkue}{\hat{\qg}_{mk}^{\UE}}
\newcommand{\gtilmkue}{\tilde{\qg}_{mk}^{\UE}}

\newcommand{\gtilmjeu}{\tilde{\qg}_{mj}^{\sn}}
\newcommand{\gamuemk}{\gamma_{mk}^{\UE}}

\newcommand{\gameumj}{\gamma_{mj}^{\sn}}

\newcommand{\betamkue}{\beta_{mk}^{\UE}}
\newcommand{\betamjeu}{\beta_{mj}^{\sn}}

\newcommand{\Rhm}{\boldsymbol{\Omega}_m}
\newcommand{\Rzj}{\boldsymbol{\Psi}_j}

\newcommand{\etamkI}{\eta_{mk}^{\mathtt{I}}}

\newcommand{\etamkpI}{\eta_{mk'}^{\mathtt{I}}}
\newcommand{\etamjE}{\eta_{mj}^{\mathtt{E}}}
\newcommand{\etamjpE}{\eta_{mj'}^{\mathtt{E}}}

\newcommand{\SINRk}{\mathrm{SINR}_k}

\newcommand{\KK}{\mathcal{K}}
\newcommand{\JJ}{\mathcal{J}}

\newcommand{\xik}{x_{\mathtt{I},k}}
\newcommand{\xikp}{x_{\mathtt{I},k'}}
\newcommand{\xej}{x_{\mathtt{E},j}}
\newcommand{\xejp}{x_{\mathtt{E},j'}}

\DeclareMathOperator{\PHI}{\boldsymbol{\Phi}}

\DeclareMathOperator{\THETA}{\boldsymbol{\Theta}}

\DeclareMathOperator{\VARPHI}{\boldsymbol{\varphi}}

\begin{document}

\title{Cell-Free Massive MIMO SWIPT with Beyond Diagonal Reconfigurable Intelligent Surfaces}

\author{Thien Duc Hua, Mohammadali Mohammadi, Hien Quoc Ngo, and  Michail Matthaiou\\
\small{
Centre for Wireless Innovation (CWI), Queen's University Belfast, U.K.\\
Email:\{dhua01, m.mohammadi, hien.ngo, m.matthaiou\}@qub.ac.uk, 
}}\normalsize
\allowdisplaybreaks

\markboth{.}%
{Shell \MakeLowercase{\textit{et al.}}: A Sample Article Using IEEEtran.cls for IEEE Journals}


\maketitle

\begin{abstract}
This paper investigates the integration of beyond-diagonal reconfigurable intelligent surfaces (BD-RISs) into cell-free massive multiple-input multiple-output (CF-mMIMO) systems, focusing on applications involving simultaneous wireless information and power transfer (SWIPT). The system supports concurrently two user groups: information users (IUs) and energy users (EUs). A BD-RIS is employed to enhance the wireless power transfer (WPT) directed towards the EUs. To comprehensively evaluate the system's performance, we present an analytical framework for the spectral efficiency (SE) of IUs and the average harvested energy (HE) of EUs in the presence of spatial correlation among the BD-RIS elements and for a non-linear energy harvesting circuit. Our findings offer important insights into the transformative potential of BD-RIS, setting the stage for the development of more efficient and effective SWIPT networks. Finally, incorporating a heuristic scattering matrix design at the BD-RIS results in a substantial improvement compared to the scenario with random scattering matrix design.
\let\thefootnote\relax\footnotetext{This work was
supported by the European Research Council
(ERC) under the European Union’s Horizon 2020 research
and innovation programme (grant agreement No. 101001331). The work of H. Q. Ngo was supported by the U.K. Research and Innovation Future Leaders Fellowships under Grant MR/X010635/1. }
\end{abstract}


\vspace{-1.2em}
\section{Introduction}
The telecommunications sector is undergoing rapid expansion, fueled by the increasing demand for high data traffic. This surge is notably driven by the deployment of 5G and 6G networks and the widespread adoption of the Internet of Everything (IoE). As a consequence, there is an anticipated substantial rise in the energy consumption of wireless networks. Recognizing this challenge, SWIPT has emerged as a promising technology to address the growing energy demands in this dynamic landscape~\cite{Clerckx:JSAC:2019}. Traditional cellular systems exhibit inherent limitations, including elevated infrastructure costs, inter-cell interference, cell-edge effects, and significant path loss attributed to extended distances. These constraints impede the seamless application of SWIPT. To tackle these challenges, CF-mMIMO emerges as a promising technology. In this paradigm, several APs are distributed across the coverage area, effectively serving multiple users. This approach reduces the distance between users and their nearby APs, fostering  macro-diversity and mitigating path loss~\cite{cite:HienNgo:cf02:2018}. As a result, CF-mMIMO provides a compelling solution to improve the SE of wireless networks and enhance the performance of WPT in SWIPT scenarios~\cite{Demir,cite:cfmmimo_swipt_Galappaththige,cite:mfmmimo_swipt_Ali}.

In recent advancements, RISs have been seamlessly integrated into wireless networks, enhancing the overall system performance. Notably, a RIS offers the capability to shape radio waves at the electromagnetic level, eliminating the need for digital signal processing methods and power amplifiers~\cite{wu2021intelligent}. Recent research efforts have delved into the exploration of RIS-assisted CF-mMIMO networks~\cite{cite:Chien:TWC:2022,Dai:TWC:2023}. Specifically, Trinh \textit{et al.}~\cite{cite:Chien:TWC:2022} conducted an analysis of the uplink and downlink SE in RIS-assisted CF-mMIMO systems.  Additionally, Dai~\textit{et al.}~\cite{Dai:TWC:2023} investigated and optimized the achievable
rate performance of uplink RIS-aided CF-mMIMO
systems.

Although the application of RISs has undeniably enhanced the WPT efficiency in various scenarios~\cite{Zhao:TCOM:2022}, their potential within the context of CF-mMIMO SWIPT systems has remained vastly unexplored. In their recent work, Shi \textit{et al.}~\cite{Shi:MCOM:2022} investigated potential application scenarios and system architectures for RIS-aided CF-mMIMO systems in the context of WPT. More precisely, they discussed the utilization of RISs in these systems to direct energy beams towards obstructed energy zones. Recently, the BD-RIS concept, which unifies different RIS modes/architectures,  has garnered significant attention from the research community~\cite{ cite:HongyuLi:BDRISoverview01:2023, Li:JSAC:2023}. This motivates us to investigate, for the first time ever, the performance of BD-RIS-assisted CF-mMIMO SWIPT systems.

We focus on a CF-mMIMO SWIPT system that serves two distinct groups: IUs and EUs. To optimize the resource utilization, we categorize the available APS into information APs (I-APs) and energy APs (E-APs). This ensures that the APs simultaneously serve IUs and EUs throughout the entire time slot, resulting in enhanced SE and EH. We employ a BD-RIS to aid in directing energy beams toward designated energy zones, thereby efficiently supporting the EUs. This novel architecture allows the EUs to harvest energy from all APs, but it also creates more interference at the IUs due to simultaneous WPT. To tackle this, we utilize the protective partial zero-forcing (PPZF) precoder, which applies partial zero-forcing (ZF) precoding at  I-APs and protective maximum ratio transmission (PMRT) at E-APS. PMRT guarantees full protection of the IUs from the energy-transmitted interference.  We consider a more realistic channel model by taking into account the spatial correlation among the BD-RIS elements. Our major contributions are:
\begin{itemize}
    \item We derive the second-order and fourth-order moments of the estimate of the AP-to-UE channel, which consists of direct and indirect links. The indirect link takes into account the spatial correlation among the channels of the BD-RIS elements. 
    \item We provide  closed-form expressions for the ergodic SE of the IUs and the HE of the EUs with non-linear energy harvester. These expressions are derived based on the channel moments and the precoding scheme. The obtained results establish a fundamental basis for designing and optimizing the performance of the considered CF-mMIMO SWIPT system.
    \item We conduct extensive simulations to compare the performance of a heuristic scattering matrix design  for BD-RIS against a random design, and no service from the BD-RIS. 
\end{itemize}

\textit{Notation:} We use bold upper/lower case letters to denote matrices/vectors. The superscripts $(\cdot)^T$ and $(\cdot)^H$ stand for the transpose and the conjugate-transpose, respectively;  $\mathbf{I}_N$ denotes the $N\times N$ identity matrix. 
A circular symmetric complex Gaussian variable with variance $\sigma^2$ is denoted by $\mathcal{CN}(0,\sigma^2)$. Finally, $\mathbb{E}\{\cdot\}$ denotes the statistical expectation.


\section{System model}~\label{sec:Sysmodel}
Figure~\ref{fig:system_model} illustrates a CF-mMIMO SWIPT system aided by a BD-RIS. In this configuration, a central processing unit (CPU) connects to $M$ APs. All APs cooperate to  simultaneously serve $K$ IUs and $J$ EUs in the same frequency bands. We define the sets $\KK\triangleq \{1,\ldots,K\}$, $\JJ\triangleq \{1,\ldots,J\}$ and $\MM\triangleq \{1,\ldots,M\}$ as the collections of indices of the IUs, EUs, and APs, respectively.
The APs are equipped with $L$ antennas each, while the IUs and EUs are equipped with a single antenna, so that $ML \gg (K+J)$. To ensure SWIPT under the same frequency spectrum, we consider the selection of operation modes for APs. Therefore, the APs are classified into I-APs and E-APs. The I-APs support wireless information transfer (WIT) to the IUs, while the E-APs support WPT towards the EUs. To further improve the HE of the EUs, a BD-RIS is located near the EU zone to support the transmission. In the BD-RIS, the $n$-th element connects to all $n'$-th elements for $n' \neq n, \forall n, n' \in \mathcal{N}$ through a network of impedances, where $\mathcal{N}=\{1,\ldots,N\}$ denotes the set of BD-RIS elements. Due to element-wise correlation,  the scattering matrix $\THETA$ is a full matrix, and satisfies $\THETA = \THETA^{T}$, and $\THETA^{H} \THETA = \textbf{I}_{N}$. In addition, the condition of perfect matching and no mutual coupling can be practically achieved by individually matching each antenna to the reference impedance and maintaining an antenna spacing greater than half-wavelength~\cite{cite:HongyuLi:BDRISoverview01:2023, Li:JSAC:2023}. 

\subsection{Channel Model}
The considered architecture operates in time division duplex, which implies that channel reciprocity holds as in the canonical form of CF-mMIMO~\cite{cite:HienNgo:cf02:2018}.
We assume a block-fading channel model, i.e., the channel remains time-invariant and frequency-flat during each coherence interval containing $\tau_{c}$ symbols, and independently varies between coherence intervals.  Particularly, we denote $\tau$ as the number of symbols per coherence interval spent on transmission of uplink training, and $(\tau_{c} - \tau)$ is the duration of downlink transmission. 

The channel vector between the $m$-th I-AP and the $k$-th IU is  denoted by $\gmkiu \in \mathbb{C}^{L \times 1}, \forall{m} \in \MM$, and $\forall{k} \in \K$. Since the BD-RIS is deployed to improve the EH system, it is located near the EU area which is far from the IU area. As a result, the indirect links reflected from the BD-RIS to the IUs are relative weak compared to the direct links, and are neglected in $\gmkiu$. The aggregated channel between the $m$-th E-AP and the $j$-th EU,  $\gmjue \in \mathbb{C}^{L \times 1}, \forall{m} \in \MM$, and $\forall{j} \in \mathcal{J}$,  is comprising the direct link, $\hmlue$, and cascaded indirect links reflected from the BD-RIS, $\Hmlue^{H} \THETA \qz_{j}$. Mathematically speaking, 
\begin{equation}~\label{eq:gmjeu}
\gmjue = \hmlue  + \Hmlue^{H} \THETA \qz_{j},
\end{equation}
where $\hmlue \in \C^{L \times 1 }$, $\Hmlue \in \C^{N \times L}$, and $\qz_{l} \in \C^{N \times 1}$ are the channels from the $m$-th AP to the j-th EU, from the $m$-th AP to the RIS, and from the RIS to the j-th EU, respectively. In particular, we denote the $l$-th column vector of $\Hmlue$ as an $N$-dimensional vector $\hmuell$, such that $\Hmlue = [\hmueo, \ldots, \hmueL]$ with $\hmuell \in \C^{N \times 1}$. We assume that $\gmkiu \sim \CN(\boldsymbol{0},\betamkue \qI_{L})$ and $\hmlue \sim \CN(\boldsymbol{0}, \betamjeu \qI_{L})$. Moreover, as in \cite{cite:Chien:TWC:2022}, we consider the spatial correlated Rayleigh fading model for the BD-RIS-related channels as follows:
\begin{align*}
& \hmuell \sim \CN (\boldsymbol{0}, \Rhm), \hspace{1em}\text{and}\hspace{1em}  \qz_{j} \sim \CN (\boldsymbol{0}, \Rzj), 
\end{align*}
where  $\Rhm = \alpha_m d_H d_V \qR$ and $\Rzj = \alpha_{j} d_H d_V \qR$, 
while $\beta^{I}_{mk}$, $\betamjeu$, $\alpha_m$, and $\alpha_j$ are the large-scale fading coefficients; $d_H$ and $d_v$ are the horizontal and vertical lengths of a RIS element. Moreover, $\qR \in \C^{N \times N}$ is the normalized spatial correlation among the RIS elements. Each element of the normalized spatial correlation matrix $\qR$, can be expressed as~\cite{cite:Chien:TWC:2022}
\begin{equation}~\label{eq:correlationRelement}
[\qR]_{nn'} = \text{sinc}\bigg( \frac{2|| \qu_{n} - \qu_{n'} ||}{\lambda} \bigg), \forall{n,n'} \in \mathcal{N},
\end{equation}
where 
\begin{equation}~\label{eq:vectoru}
\qu_{n} \!=\! [0, \text{mod}(n\! -\! 1) N_H,\! \lfloor (n-1)/N_H \rfloor N_V]^T, \forall{n}\in \mathcal{N},
\end{equation}
is the location vector of the $n_I$-th element w.r.t. the origin. In \eqref{eq:vectoru}, $N_H$ and $N_V$ are the respective numbers of RIS elements per row and per column, by which $N = N_{H} N_{V}$.

Given the complex analytical nature of the aggregated channel, we are providing the following Lemma to facilitate subsequent derivations.

\begin{Lemma}~\label{Lemma:2ndand4thmoment}
The second-order and fourth-order moments of the channel norm $\gmjue$ can be written as
\begin{subequations}
  \begin{align}
        \Ex\{ \Vert \gmjue \Vert^2 \} &=  L\delta_{mj}, \\
        \Ex\{ \Vert \gmjue \Vert^4 \} &= L(L+1)\big( \delta_{mj}^2 + \trace((\bar{\THETA}_{mj})^2) \big), ~\label{eq:2ndand4thmoment}
\end{align}  
\end{subequations}
where  $\delta_{mj} \triangleq \betamjeu + \trace(\bar{\THETA}_{mj})$, and $\bar{\THETA}_{mj} \triangleq \THETA^{H} \Rhm \THETA \Rzj$. 
\end{Lemma}

\begin{proof}
    See Appendix~\ref{appendix:A}.
\end{proof}

\begin{figure}[t]
	\centering
	\vspace{0em}
	\includegraphics[width=78mm]{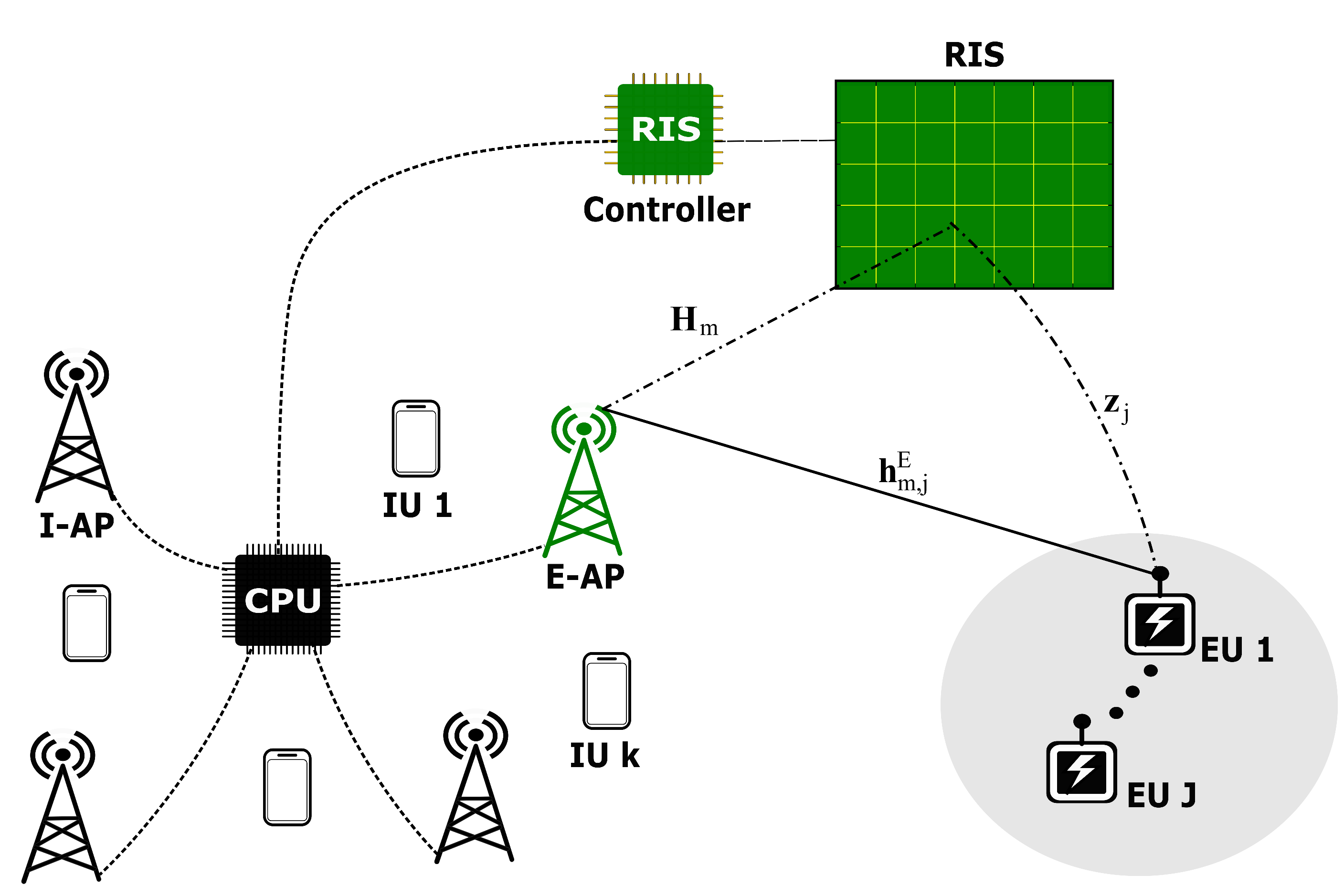}
	\vspace{0em}
	\caption{The proposed  CF-mMIMO RIS-assisted SWIPT system.}
	\vspace{0.8em}
	\label{fig:system_model}
\end{figure}

\vspace{-1.5em}
\subsection{Uplink Training for Channel Estimation}\label{phase:ULforCE}
In the training phase, the channels are estimated at the APs. To this end, all the IUs and EUs transmit orthogonal pilot sequences $\VARPHI_{k}^{I}$ and $\VARPHI_{j}^{E}$ of length $\tau$ symbols, which requires $\tau \geq \vert \mathcal{K} \vert + \vert \mathcal{J} \vert$. Let $\PHI_{K} = [\VARPHI_{1}^{I},\ldots,\VARPHI_{K}^{I}]$ and $\PHI_{J} = [\VARPHI_{1}^{E},\ldots,\VARPHI_{J}^{E}]$ denote the pilot matrices transmitted by the IUs and EUs. By applying similar analytical methodology as in \cite{cite:HienNgo:cf01:2017}, the minimum mean square error (MMSE) estimates of $\gmkiu$ and $\gmjue$ are $\hgmkue \sim \CN(\boldsymbol{0},\gamuemk \qI_{L})$ and $\hgmjue \sim \CN(\boldsymbol{0},\gameumj \qI_{L})$,  where 
\begin{align}~\label{eq:gammaI_APIU}
&\gamuemk \triangleq \Ex \Big\{ \big\Vert [\hgmkue]_{l} \big\Vert^{2} \Big\} = \frac {\tau \rho_{u}(\betamkue)^2} {\tau \rho_{u} \betamkue + 1}, \\
\vspace{-2em}
&\gameumj \triangleq \Ex \Big\{ \big\Vert [\hgmjue]_{l} \big\Vert^{2} \Big\} = \frac{\tau \rho_{u}(\delta_{mj})^2} {\tau \rho_{u} \delta_{mj} + 1},
\end{align}
 where $\rho_{u}$ denotes the uplink training signal-to-noise ratio (SNR). Due to the independence property of MMSE estimation, the estimation errors are distributed as $\gtilmkue \sim \CN(\boldsymbol{0}, (\betamkue - \gamuemk) \qI_L)$ and $\gtilmjeu \sim \CN(\boldsymbol{0}, (\delta_{mj} - \gameumj) \qI_L)$.
\vspace{-0.5em}
\subsection{Downlink SWIPT}
During downlink transmission phase, the $K$ IUs and $J$ EUs are ensured simultaneous downlink service. Following the channel acquisition from the training phase, $M$ APs initially select between the WIT or WPT operation mode. Subsequently, according to the mode of operation, precoding schemes are appropriately employed for WIT and WPT operations. 
Specifically, the decision on which operation mode is assigned to each AP is made based on the network design requirements. 
We denote the binary variable $a_m$ as the operation mode indicator, which can be formulated as
\begin{align}~\label{eq:mode_indicator}
  a_m 
  &\triangleq 
    \begin{cases}
      1 & \text{if the $m$-th AP operates as I-AP}\\
      0 & \text{if the $m$-th AP operates as E-AP}.
    \end{cases}       
\end{align}
We denote $\xik$ and $\xej$, where $\Ex\{ \vert \xik \vert^2 \} = \Ex\{ \vert \xej \vert^2 \} =1$, as the data and energy symbol transmitted to the $k$-th IU and $j$-th EU, respectively.  Then, the transmitted signal from the $m$-th AP is 
\vspace{-0.2em}
\begin{align}~\label{eq:x_m}
\qx_{m} 
&= \sqrt{a_m\rho_{d}}\sum\nolimits_{k\in\KK}\sqrt{\etamkI} \wimk \xik \nonumber \\
&+  \sqrt{(1-a_m)\rho_{d}}\sum\nolimits_{j\in\JJ} \sqrt{\etamjE}\wemj \xej,
\end{align}
where $\rho_{d} = \tilde{\rho_{d}}/\Snn$ is the maximum downlink SNR, while $\tilde{\rho_{d}}$ denotes the transmit power and $\Snn$ is the noise power; $\wimk \in \C^{L\times 1}$ and $\wemj\in \C^{L\times 1}$ are the precoding vectors for the $k$-th IU and $j$-th EU, respectively, which adhere to the constraint $\Ex\big\{\big\Vert\wimk\big\Vert^2\big\}=\Ex\big\{\big\Vert\wemj\big\Vert^2\big\}=1$. We denote the transmitted signal from the $m$-th I-AP and the $m$-th E-AP as $\qx_{\mathtt{I},m}$ and $\qx_{\mathtt{E},m}$, respectively. Furthermore, $\etamkI$ and $\etamjE$ are the  power control coefficients chosen to satisfy the power constraint at each AP, i.e.,
\begin{align}
& a_m\Ex\big\{\big\Vert \qx_{\mathtt{I},m}\big\Vert^2\big\}+ (1-a_m)\Ex\big\{\big\Vert \qx_{\mathtt{E},m}\big\Vert^2\big\}\leq \rho_{d}.
\end{align}
With the transmitted signal $\qx_m$, the $k$-th IU and the $j$-th EU then receive
\begin{align}~\label{eq:yik}
   &\yik 
    \!=\! \sqrt{a_m \rho_d} \sum\nolimits_{m\in\MM} \sqrt{\etamkI} (\gmkiu)^{H} \wimk \xik    \\
    &+ \sqrt{a_m \rho_d} \sum\nolimits_{k'\neq k}^{K}\sum\nolimits_{m\in\MM} \sqrt{\etamkpI} (\gmkiu)^{H} \wimkp \xikp \nonumber \\
    &+\! \sqrt{(\!1\!-\!a_m\!) \rho_d}\! \sum\nolimits_{j\in\JJ}\!\!\sum\nolimits_{m\in\MM}\!\! \sqrt{\etamjE} (\gmkiu)^{\!H}\! \wemj \xej  \!\!+\! n_{d,k}, \nonumber 
\end{align}
and,
\begin{align}~\label{eq:yel}
    &\yej 
    \!=\! \sqrt{(1\!-\!a_m) \rho_d} \!\sum\nolimits_{j'\in\JJ}\!\!\sum\nolimits_{m\in\MM}\!\! \sqrt{\etamjpE} (\gmjue)^{\!H} \wemjp \xejp   \nonumber \\
    &+\! \sqrt{a_m \rho_d}\! \sum\nolimits_{k\in\KK}\!\sum\nolimits_{m\in\MM}\!\! \sqrt{\etamkI} (\gmjue)^{\!H} \wimk \xik \! +\! n_{d,j}, 
\end{align}
respectively, where $n_{d,k}$ and $n_{d,j} \sim \CN(0,1)$ are the  corresponding additive white Gaussian noises.

\section{Precoding Design and Performance Evaluation}

\subsection{Precoding Design}

Different beamforming designs can be explored at the system design level to meet the performance requirements of the system. In order to strike a balance between complexity and performance, we consider the PZF technique. Then,  local PZF precoding is considered at the I-APs and PMRT is applied at the E-APs. This precoding design capitalizes on the effectiveness of the ZF principle in suppressing inter-user interference, making it nearly optimal for information transmission~\cite{cite:pzf_pmrt_2020}. Additionally, power transfer with MRT proves to be optimal for WPT systems, especially with a large number of antennas~\cite{cite:MRT_for_HE:Almradi}. However, it is worth noting that in our system, IUs also encounter non-coherent interference, denoted as $\mathrm{EUI}_{kj}$, from EUs due to the simultaneous transmission of information and energy signals.
To mitigate the non-coherent interference experienced from energy signals sent to EUs, MRT can be implemented in the orthogonal complement of the IUs’ channel space, which is called PMRT. Let $\Ghmu = \big[\ghmonue, \ldots, \ghmKue\big] \in \C^{L\times K}$ and $\Ghms= \big[\ghmons, \ldots, \ghmJue\big] \in \C^{L\times J}$ be  the matrices of the estimated channels between the $m$-th AP and all IUs, and all EUs, respectively. Accordingly, the PZF and PMRT precoders at the $m$-th AP towards the $k$-th IU and the $j$-th EU can be designed as
\begin{subequations}
 \begin{align}
    \wimk^{\PZF} &=\alpha_{mk}^{\PZF}{ \Ghmu \Big(\big(\Ghmu\big)^H \Ghmu\Big)^{-1} \qe_k^I}
    ,~\label{eq:wipzf}\\
        \wemj^{\PMRT} &= \alpha_{mj}^{\PMRT}{\qB_m\Ghms\qe_{j}^{E}},~\label{eq:wemrt}
\end{align}   
\end{subequations}
where 
\begin{align}
    \alpha_{mk}^{\PZF}&\triangleq\Big(\Ex \Big\{ \big\Vert \Ghmu \big(\big(\Ghmu\big)^H \Ghmu\big)^{-1} \qe_k^I  \big\Vert^2 \Big\}  \Big)^{-\frac{1}{2}}
    \nonumber\\
    {\alpha_{mj}^{\PMRT}}&\triangleq\Big( \Ex \big\{ \big\Vert \qB_m\Ghms\qe_{j}^{E} \big\Vert^2\big\} \Big)^{-\frac{1}{2}},
\end{align}
with $\qe_{k}^{I}$ and $\qe_{j}^{E}$ being the $k$-th column of $\qI_{K}$ and the $j$-th column of $\qI_{J}$, respectively. In addition, $\qB_m$ denotes the projection matrix onto the orthogonal complement of $\Ghmu$ so that $\big(\ghmkue\big)^H \qB_m =\boldsymbol{0}$. Thus, $\qB_m$ can be computed as
\begin{align}
  \qB_m  = \qI_{L}  - \Ghmu \Big( \big(\Ghmu\big)^H \Ghmu\Big)^{-1}  \big(\Ghmu\big)^H.
\end{align}

According to~\cite{cite:pzf_pmrt_2020}, under the consideration of independent Rayleigh fading channels, the analytical terms of the normalization factors in \eqref{eq:wipzf} and \eqref{eq:wemrt} can be calculated as $\big(\alpha_{mk}^{\PZF}\big)^2= {(L-K)\gamuemk}$ and $\big(\alpha_{mj}^{\PMRT}\big)^2= \big((L-K)\gameumj \big)^{-1}$.

\vspace{-1em}
\subsection{ Spectral Efficiency and Average Harvested Energy}
Upon receiving the signal \eqref{eq:yik}, the $k$-th IU detects its desired symbol $\xik$. In the absence of a downlink training phase, the IUs avail of channel statistics, specifically the average effective channel gain, to detect their desired symbols. For the downlink SE analysis at the IUs, we apply the hardening bound~\cite{cite:HienNgo:cf01:2017}. Consequently, the received signal at the $j$-th IU can be expressed as
\begin{align}~\label{eq:yi:hardening}
    \yik &=  \mathrm{DS}_k  \xik +
    \mathrm{BU}_k \xik 
         +\sum\nolimits_{k'\in\K \setminus k}
     \mathrm{IUI}_{kk'}
     \xikp
    + \sum\nolimits_{j\in\J}
     \mathrm{EUI}_{kj}\xej + n_k,~\forall k\in\K,
\end{align}
where $\mathrm{DS}_k$, $\mathrm{BU}_k$, $\mathrm{IUI}_{kk'}$, and $\mathrm{EUI}_{kj}$ represent the desired signal, the beamforming gain uncertainty, the interference cause by the $k'$-th IU, and the interference caused by the $j$-th EU, respectively, given by
\begin{subequations}
  \begin{align}~\label{eq:yi:components}
\mathrm{DS}_k &\triangleq \sum\nolimits_{m\in\MM} \sqrt{a_m \rho_d \etamkI} \Ex \big\{(\gmkiu)^{H} \wimk^{\PZF} \big\},  \\
\mathrm{BU}_k &\triangleq \sum\nolimits_{m\in\MM} \sqrt{a_m \rho_d \etamkI} \Bigl( (\gmkiu)^{H} \wimk^{\PZF} - \Ex \big\{(\gmkiu)^{H} \wimk^{\PZF} \big\} \Bigl)~\label{eq:component_BUk}, \\
\mathrm{IUI}_{kk'} &\triangleq \sum\nolimits_{m\in\MM} \sqrt{a_m \rho_d\etamkpI} (\gmkiu)^{H} \wimkp^{\PZF}, \\
\mathrm{EUI}_{kj} &\triangleq \sum\nolimits_{m\in\MM} \sqrt{(1-a_m) \rho_d\etamjE} (\gmkiu)^{H} \wemj^{\PMRT}.  
\end{align}  
\end{subequations}

The corresponding DL SE in [bit/s/Hz] for the $k$-th IU can be obtained as
\begin{align}~\label{eq:SEk:Ex}
    \mathrm{SE}_k
      &=
      \Big(1\!- \!\frac{\tau}{\tau_c}\Big)
      \log_2
      \left(
       1\! + \SINRk
     \right),
\end{align}
where $\SINRk$ is the effective signal-
to-interference-and-noise (SINR), given by $\SINRk=$ 
\begin{align}~\label{eq:SINE:general}
       &\!\frac{
                 \big\vert  \mathrm{DS}_k  \big\vert^2
                 }
                 {  
                 \Ex\big\{ \big\vert  \mathrm{BU}_k  \big\vert^2\!\big\} \!+
                  \!\!
                 \sum_{k'\in\KK\setminus k}\!
                  \Ex\big\{ \big\vert \mathrm{IUI}_{kk'} \big\vert^2\!\big\}
                  \! + \!
                   \!
                  \sum_{j\in\JJ}\!
                 \Ex \big\{ \big\vert  \mathrm{EUI}_{kj} \big\vert^2\!\big\}
                   \!+\!  1}. 
\end{align}

{{To characterize the HE, a non-linear energy harvesting model with the sigmoidal function is used. Therefore, the total HE at EU $j$ is given by~\cite{Boshkovska:CLET:2015}
 \vspace{0.1em}
  \begin{align}~\label{eq:NLEH}
  \Phi_j\big(\qa,  \ETAE,\ETAI\big) = \frac{\Lambda_{j}\big(\mathrm{E}_{j}(\qa,  \ETAE, \ETAI)\big) - \phi \nu }{1-\nu}, ~\forall j\in\JJ,
 \end{align}
 where $\qa$ is an indicator vector, whose entries are $a_m$; $\ETAI = [\eta_{m1}^{\mathtt{I}}, \ldots, \eta_{mK}^{\mathtt{I}}]$; $ \ETAE = [\eta_{m1}^{\mathtt{E}}, \ldots, \eta_{mJ}^{\mathtt{E}}]$;  $\phi$ is the maximum output DC power; $\nu=\frac{1}{1 + \exp(\xi \chi)}$ is a constant to guarantee a zero input/output response, while $\xi$ and $ \chi$ are constant related parameters that depend on the circuit. Moreover, $\Lambda\big(\mathrm{E}_{j}(\qa,  \ETAE, \ETAI)\big)$  is the traditional logistic function, given by 
  \begin{align}~\label{eq:PsiEl}
     \Lambda_{j}\big(\mathrm{E}_{j}(\qa,  \ETAE,  \ETAI)\big) &\!=\!\!\frac{\phi}{1 \!+ \!\exp\big(\!-\xi\big(\mathrm{E}_{j}(\qa, \ETAE,\ETAI)\!-\! \chi\big)\!\big)}.
 \end{align}
Note that $\mathrm{E}_{j}(\qa, \ETAE,\ETAI)$ denotes the received RF energy at EU $j$, $\forall j\in\JJ$. We consider the average of the harvested energy as the performance metric of the WPT operation, given by
  \begin{align}~\label{eq:NLEH:av}
  \Ex\big\{\Phi_j\big(\qa,  \ETAE,\ETAI\big)\big\} = \frac{\Ex\big\{\Lambda_{j}\big(\mathrm{E}_{j}(\qa,  \ETAE, \ETAI)\big)\big\} - \phi \nu }{1-\nu}.
 \end{align}
Finding $\Ex\big\{\Lambda_{j}\big(\mathrm{E}_{j}(\qa,  \ETAE, \ETAI)\big)\big\}$ is complicated if not impossible. Nevertheless, since the logistic function in~\eqref{eq:PsiEl} is a convex function of $\mathrm{E}_{j}(\qa, \ETAI, \ETAE)$, by using Jensen's inequality, we have
\vspace{-0.6em}
\begin{align}\label{eq:Jensen}
 \Ex\left\{\Lambda_{j}\left(\mathrm{E}_{j}\big(\qa,  \ETAI, \ETAE\big)\right)\right\}&
 \geq \Lambda_{j}\left(\Ex\left\{\mathrm{E}_{j}\big(\qa, \ETAI, \ETAE\big)\right\}\right) 
 \nonumber\\
 &=\Lambda_{j}\left(Q_{j}\big(\qa,  \ETAI, \ETAE\big)\right),
 \end{align}
where $Q_{j}\big(\qa,  \ETAI, \ETAE\big) = \Ex\{\mathrm{E}_{j}(\qa,  \ETAI, \ETAE\big)\}$. Thus, we can find a tight lower-bound on the average HE at the EU $j\in\JJ$. By using~\eqref{eq:yel}, we have
\vspace{0.1em}
\begin{align}~\label{eq:El_average}
     &\Ex\big\{\mathrm{E}_{j}(\qa,  \ETAI, \ETAE\big)\big\} =(\tau_c-\tau)\Snn 
     \bigg( {\rho}_d\!\sum\nolimits_{m\in\MM}\!\!
   {(1\!-\!a_m)\etamjE} \Ex\Big\{\big\vert\big(\gmjue\big)^{\!H}\wemj^{\PMRT}\big\vert^2\Big\} \nonumber \\
     &\hspace{-1.5em}+{\rho}_d\!\!\sum_{j' \in\JJ\setminus j}\!\sum\nolimits_{m\in\MM}\!\!
   {(1\!-\!a_m)\etamjpE} \Ex\Big\{\big\vert\big(\gmjue\big)^{\!H}\wemjp^{\PMRT}\big\vert^2\Big\} 
    +\!{\rho}_d
   \sum\nolimits_{k\in\KK}\!\sum\nolimits_{m\in\MM}\!
   {a_m\etamkI}\Ex\Big\{\big\vert\big(\gmjue\big)^{\!H}\wimk^{\PZF}\big\vert^2\Big\} \!+\! 1 \bigg).
\end{align}

We provide closed-form expressions for the SE and average HE under the PPZF precoding scheme. 

\begin{proposition}~\label{Theorem:SE:PPZF}
The ergodic SE for the $k$-th IU, achieved by PZF precoding at the I-APs and PMRT at the E-APs is given by~\eqref{eq:SEk:Ex}, where the effective SINR is given in closed-form by~\eqref{eq:SINE:PPZF} at the top of next page.
\begin{figure*}
\begin{align}~\label{eq:SINE:PPZF}
    &\SINRk^{\PZF}\big(\qa,  \ETAI, \ETAE\big) =
    \!\frac{
                   (L-K)\Big(\sum_{m\in\MM}\sqrt{ a_m\etamkI \gamuemk}  \Big)^2
                 }
                 { 
                  \sum_{k'\in\mathcal{K}} \sum_{m\in\MM} a_m \etamkpI (\betamkue  -  \gamuemk) 
                 + 
                  \sum_{j\in\mathcal{J}} \sum_{m\in\MM} {(1 - a_m)\etamjE(\betamkue - \gamuemk)}
                   +  1/\rho_{d}}.
\end{align}
	\hrulefill
	\vspace{-4mm}
\end{figure*}
\end{proposition}
\begin{proof}
   The proof is omitted due to space limitations.
\end{proof}

\begin{proposition}~\label{Theorem:RF:PPZF}
The average HE for the $j$-th EU, achieved by PMRT precoding at the E-APs and PZF at the I-APs is bounded by
 \vspace{-0.5em}
  \begin{align}~\label{eq:NLEH:av_bound}
  \Ex\big\{\Phi_j\big(\qa,  \ETAE,\ETAI\big)\big\} \geq \frac{\Lambda_{j}\left(Q_{j}\big(\qa,  \ETAI, \ETAE\big)\right) - \phi \nu }{1-\nu},
 \end{align}
where $Q_{j}\big(\qa,  \ETAI, \ETAE\big)$ is given by ~\eqref{eq:El_average:PPZF} at the top of next page.
\vspace{-1em}
\begin{figure*}
\begin{align}~\label{eq:El_average:PPZF}
   Q_{j}\big(\qa,  \ETAI, \ETAE\big) &=
     (\tau_c - \tau)\Snn \rho_{d}
     \Big(
     \big(L-K+1\big) \sum\nolimits_{m \in \MM}
     (1-a_m) \etamjE \gameumj
        +
   \sum\nolimits_{j'\in\mathcal{J} \setminus j} \sum\nolimits_{m \in \MM}
   (1 - a_m)\etamjpE \delta_{mj}
   \nonumber\\
    &\hspace{8em}+
   \sum\nolimits_{k\in\mathcal{K}} \sum\nolimits_{m \in \MM}
   a_m \etamkI \delta_{mj} + 1 / \rho_{d} \Big).
\end{align}
	\hrulefill
	\vspace{-5mm}
\end{figure*}
\end{proposition}

\begin{proof}
   See Appendix~\ref{appendixHE}.
\end{proof}

By examining the expressions for the SE~\eqref{eq:SINE:PPZF} and HE~\eqref{eq:NLEH:av_bound}, we note that utilizing the current degrees-of-freedom for system design, encompassing the design of the scattering matrix for the BD-RIS, the AP mode selection vector, and the power control coefficients at the AP, enables us to fulfill various design objectives in CF-mMIMO SWIPT systems.

\begin{figure}[t]
	\centering
	\vspace{-0.75em}
	\includegraphics[width=80mm]{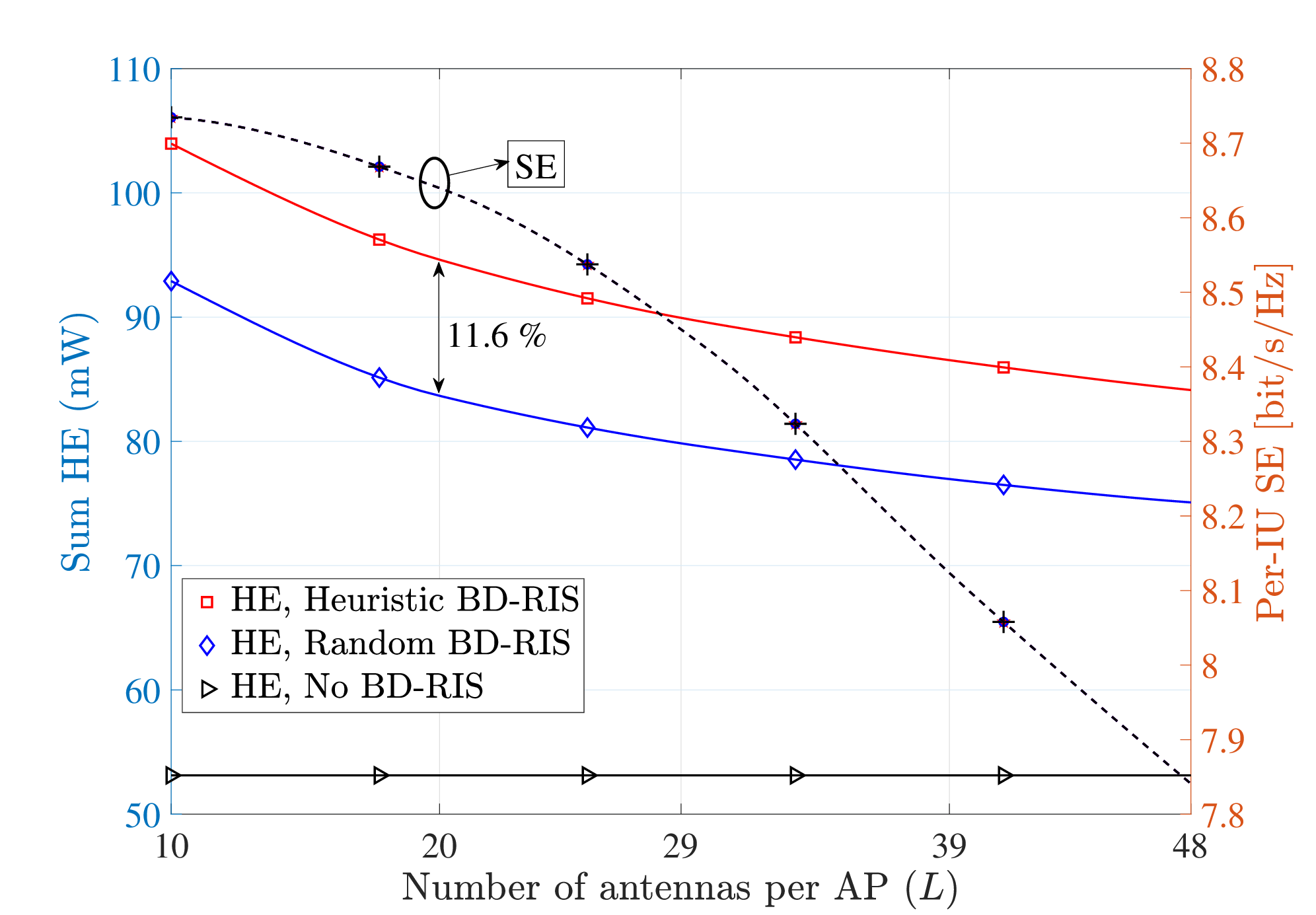}
	\vspace{-0.6em}
	\caption{Impact of the number of antennas per AP ($L$) on the average HE and SE ($N=40, ML = 480, K = 3, J = 5$).}
	\vspace{0.8em}
	\label{fig:Fig2}
\end{figure}

\vspace{-0.6em}
\section{Numerical Results and Discussion}~\label{sec:simulation}
In this section, numerical results are presented to illustrate the performance of our proposed scenario. We consider a BD-RIS-assisted CF-mMIMO SWIPT  network, where the APs are uniformly distributed in an area of $1 \times 1 \text{km}^2$. The height of the APs, BD-RIS, and users is $15$ m, $30$ m, and $1.65$ m, respectively. Moreover, we set $\tau_c = 200$, $\tau = K + J$, $\Snn = -92$ dBm, $\rho_d = 1$ W and $\rho_u = 0.2$ W.  
{In addition, we set the non-linear energy harvesting parameters as $\xi = 150$, $\chi = 0.014$, and $\phi = 0.024$ W \cite{cite:MRT_for_HE:Almradi}}. 
The large-scale fading coefficients are generated following the three-slope propagation model from \cite{cite:HienNgo:cf01:2017}. The shadow fading follows a log-normal distribution with a standard deviation of $8$ dB. In our simulations, we compare three different design scenarios: 

\begin{itemize}
    \item Heuristic scattering matrix design at the BD-RIS (\textbf{Heuristic BD-RIS}), where the symmetric-unitary projection scheme in~ \cite{YMao_BDRIS_2023} is adopted. {To this end, we first construct $M$ matrices $\qF_m \triangleq \Hmlue \qH^{E}_m \qZ_{m}^{H}$, where $\qH^{E}_m =[\hmlueo,\ldots,\hmlueJ] \in \C^{M \times J}$ and $\qZ_{m} = [\qz_1,\ldots,\qz_J] \in \C^{N \times J}$. Next, the symmetric projection of $\qF_m$ is obtained by averaging $\qF_m$ and its transpose across various small-scale channel realizations. To meet the unitary constraint, the rank $r_m$ and singular value decompositions $[\qU_m, \qV_m]$ of $\qF_m$ are computed. Then, $\qU_m, \qV_m$ are partitioned into $[\qU_{m}^{r_m}, \qU_{m}^{N-r_m}]$, and $[\qV_{m}^{r_m}, \qV_{m}^{N-r_m}]$, respectively. This approach allows for the heuristic design of $\THETA$ via the symmetric-unitary projection, which can be computed as $\THETA_{m} \triangleq \hat{\qU}_m \qV^{H}$, where $\hat{\qU}_m = [\qU_{m}^{r_m}, \big( \qV_{m}^{N-r_m} \big)^{*}]$. The trace products $\trace(\THETA^{H} \qR \THETA \qR)$ and the indexes of $\THETA_{m}$ of each iteration are stored, such that
    \begin{align*}
        \THETA_{m}^{*} =\underset{\THETA_{m}}{\argmax} {\trace(\THETA^{H} \qR \THETA \qR)}.
    \end{align*}
    This process iteratively refines the design of the scattering matrix at the BD-RIS.}
    \item BD-RIS with a random scattering matrix (\textbf{Random BD-RIS}), where a discrete Fourier transformation matrix is used to generate random $\THETA$ matrices, which are then normalized by $N$.
    \item CF-mMIMO SWIPT without BD-RIS (\textbf{No BD-RIS}), to study the role of BD-RIS in the SWIPT process.    
\end{itemize}
For both Heuristic BD-RIS and Random BD-RIS, the scattering matrix satisfies the constraints $\THETA = \THETA^{T}$, and $\THETA^{H} \THETA = \mathbf{I}_{N}$~\cite{ cite:HongyuLi:BDRISoverview01:2023, Li:JSAC:2023}. We assume that the APs' operation mode selection parameters ($a_m$, $\forall m\in \MM$) are randomly assigned, while the downlink power control coefficients are $\etamkI = \big(\sum_{k\in\KK}\gamuemk\big)^{-1}$, $\forall m, k$,  and $\etamjE = \big(\sum_{j\in\JJ}\gameumj\big)^{-1}$, $\forall m, j$.



Figure~\ref{fig:Fig2} shows the average sum HE and the per-IU SE achieved by three scenarios versus the number of antennas per AP. We note that, for a fixed number of service antennas (i.e., $ML = 480$), the number of APs decreases as the number of antennas per AP $L$ increases. From the standpoint of average sum heuristic HE, it is evident that Heuristic BD-RIS outperforms Random BD-RIS, achieving a noticeable gain of up to $12\%$. This result showcases the significance of an optimization scheme for BD-RIS scattering matrix design. On the other hand, it is evident that without BD-RIS, WPT fails to meet the energy requirements of the EUs. Finally, we observe that, by increasing $L$ and consequently decreasing $M$, both the per-IU SE and average sum HE are degraded. The decrease rate of the SE is more dominant compared to the decrease rate of the average sum HE, which can be interpreted as a consequence of employing BD-RISs to assist the EUs.

Figure~\ref{fig:Fig3} shows the average HE of the EUs and per-IU SE versus the number of APs for fixed $ML = 480$. As the number of APs increases from $M=8$ to $M=24$, there is a notable enhancement in the average HE, with improvements of $14\%$ and $20\%$ for Random BD-RIS and Heuristic BD-RIS, respectively. Beyond this, the improvement continues, albeit at a more gradual rate. These results confirm that both the number of APs and the per-AP antenna count should be carefully selected to improve both SE and average HE in the considered CF-mMIMO SWIPT system. 


\begin{figure}[t]
	\centering
	\vspace{-.75em}
	\includegraphics[width=80mm]{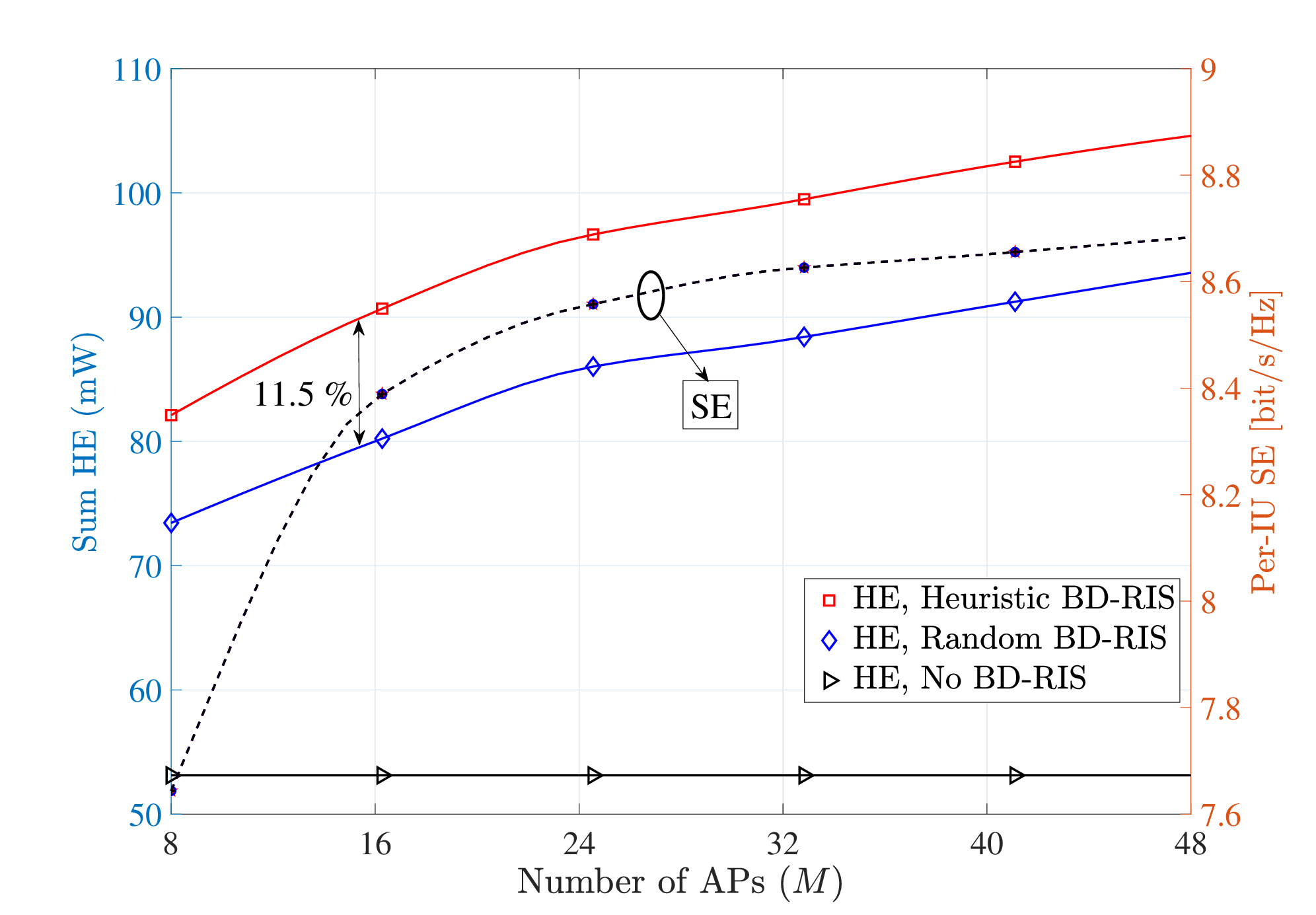}
	\vspace{-0.6em}
	\caption{Impact of the number of APs ($M$) on the average HE and SE ($N=40, ML = 480, K = 3, J = 5$).}
	\vspace{0.8em}
	\label{fig:Fig3}
\end{figure}

\vspace{-0.6em}
\section{Conclusion}~\label{sec:conclusion}
We evaluated the SE and average HE performance in a novel setup of a CF-mMIMO SWIPT network assisted by a BD-RIS. In this setup, the APs function as either I-APs or E-APs, serving two distinct groups of IUs and EUs simultaneously, with the BD-RIS strategically positioned near the energy zone to facilitate the WPT towards the EUs. Our simulation results demonstrated the superior performance of the BD-RIS assisted architectures compared to the conventional designs without BD-RIS. We further observed that the proposed Heuristic BD-RIS can achieve up $12\%$ gain over the Random BD-RIS design for a fixed number of service antennas. 
For future work, integrating a robust optimization scheme to jointly optimize the design variables could pave the way for next-generation communication systems.


\appendices

\vspace{-0.8em}
\section{Proof of Lemma~\ref{Lemma:2ndand4thmoment}}\label{appendix:A}
The following lemma is useful to prove Lemma~\ref{Lemma:2ndand4thmoment}. The proof is omitted due to space limitations. 
\begin{Lemma}~\label{lemma:EXMXXNX}
    For an $N \times L$ matrix $\qX = [\qx_1,...,\qx_L]$ with $\qx_l \in \C^{N\times1}$ distributed as $\qx_l \sim \CN(\boldsymbol{0},\bar{\qR})$ with $\qR \in \C^{N\times N}$, and two deterministic matrices $\qM$, $\qN \in \C^{N \times N}$, it holds that $\Ex\bigl\{ \qX^{H} \qM \qX \qX^{H} \qN \qX \bigl\}$ is a diagonal matrix of size $L\times L$ with $L \trace\big( \qM  \qR  \qN \qR\big)
         + \trace\big( \qR\qM  \big) \trace(  \qR\qN \big)$ on its main diagonal. 
\end{Lemma}


By leveraging the statistical independence of the direct and indirect links, the second-order moment of $\gmjue$ becomes
\begin{align}
    \Ex\{ \Vert\gmjue\Vert^2 \} 
    &= \Ex{\Vert\hmlue\Vert^2} + \Ex\Big\{ \Vert \Hmlue^{H} \THETA \qz_{j} \Vert^2 \Big\} \nonumber \\
    &= L \betamjeu + \Ex \Big\{\qz^{H}_{j} \THETA^{H} \Hmlue \Hmlue^{H} \THETA \qz_{j} \Big\}
    \nonumber \\
    &\overset{\mathrm{(a)}}{=} L \betamjeu + \trace\Big( \THETA^{H} \Ex \big\{ \Hmlue \Hmlue^{H}\big\} \THETA \Ex \big\{ \qz_{j} \qz^{H}_{j}\big\} \Big) 
    \nonumber \\
    &= L \Big( \betamjeu + \trace\Big( \THETA^{H} \Rhm \THETA \Rzj \Big)\Big),
\end{align}
where (a) exploits $\Ex\{\qx^{H} \qy\} = \Ex\{\trace(\qy \qx^{H})\}$, followed by the property of expectation of independent channels $\qH_m$ and $\qz_{j}$.

To derive the fourth-order moment $\Ex\{ \Vert\gmjue\Vert^4 \}$, we define  $a \!=\!  \Vert \hmlue \Vert^2$, $
b \!=\! (\hmlue)^H \Hmlue^{H} \THETA \qz_{j}$, $
c \!=\! \qz_j^{H} \THETA^H \Hmlue \hmlue$, and $d \!=\!\Vert \Hmlue^{H} \THETA \qz_{j} \Vert^2$. Thus, we have
\begin{align}~\label{eq:Ex_gmjue4}
    \Ex\{ \Vert\gmjue\Vert^4 \} 
    \!= &\Ex\{\vert a \vert^2\} \!+\! \Ex\{| b |^2\} \!+\! \Ex\{| c |^2\}\!+\!\Ex\{| d |^2\}\!+\!2 \Ex\{ad\}.
\end{align}
We notice that
\begin{equation}~\label{eq:Ex_a^2}
    \Ex\{\vert a \vert^2\} = \Ex\{\Vert \hmlue \Vert^4\} = L(L+1)(\betamjeu)^2.
\end{equation}
Moreover, $\Ex\{\vert b \vert^2\}=\Ex\{\vert c \vert^2\}$, which can be calculated as
\begin{align}~\label{eq:Ex_b^2}
   \Ex\{| b |^2\} 
   &= \betamjeu \trace( \THETA^H  \Ex \{ \Hmlue \Hmlue^{H} \} \THETA \Ex\{ \qz_{j}\qz_j^{H} \})  \nonumber \\
   &= L \betamjeu\bigl( \trace( \THETA^H \Rhm \THETA \Rzj) \bigl).
\end{align}
In order to derive $\Ex\big\{\vert d \vert^2 \big\}$, we first express it as
\begin{align}~\label{eq:Ex_d^2}
   \Ex\big\{\vert d \vert^2 \big\} 
    &= \Ex\Big\{ \trace(\Rzj \THETA^{H} \Hmlue \Hmlue^{H} \THETA \Rzj \THETA^{H} \Hmlue \Hmlue^{H} \THETA) \nonumber \\  
    &\hspace{1em}+ \big\vert \trace(\Rzj \THETA^{H} \Hmlue \Hmlue^{H} \THETA) \big\vert^2   \Big\},
\end{align}
where \cite[Lemma 2]{cite:Michail:Lemma2:2015} is applied for a complex Gaussian random vector $\qz_{j} \sim \CN(\boldsymbol{0},\Rzj)$ and a given deterministic matrix $\THETA^{H} \Hmlue \Hmlue^{H} \THETA $. 

The first expectation term in~\eqref{eq:Ex_d^2} can be obtained as
\begin{align}~\label{eq:Ex_ls}
    &\Ex \Big\{ \trace\big(\Hmlue^{H} \THETA \Rzj \THETA^{H} \Hmlue \Hmlue^{H} \THETA \Rzj \THETA^{H} \Hmlue \big)   \Big\} \nonumber \\
    &= \trace \Big(\Ex \Big\{ \Hmlue^{H} \THETA \Rzj \THETA^{H} \Hmlue \Hmlue^{H} \THETA \Rzj \THETA^{H} \Hmlue \Big\} \Big) \nonumber \\
    &\!\stackrel{(a)}{=} \!
    L\bigg(L \trace\big( \THETA \Rzj \THETA^{H}  \Rhm  \THETA \Rzj \THETA^{H} \Rhm\big)
    \nonumber\\
    &\hspace{2em}
         + \trace\big( \Rhm\THETA \Rzj \THETA^{H}  \big) 
         \trace(  \Rhm\THETA \Rzj \THETA^{H} \big)\Big),
         \nonumber\\
     &\!=\! 
    L\Big(L \trace\big(( \THETA \Rzj \THETA^{H}  \Rhm  )^{\!2}\big)
             \!+ \! \big(\trace( \Rhm\THETA \Rzj \THETA^{\!H} ) \big)^{\!2}\Big),
\end{align}
where we have used  Lemma~\ref{lemma:EXMXXNX}.

Due to space limitations, we present the conclusive results for the second terms in~\eqref{eq:Ex_d^2} as
  \begin{align}
 &\Ex\Big\{ \Big\vert \trace\big( \Hmlue^{H} \THETA \Rzj \THETA^{H} \Hmlue \big) \Big\vert^{2} \Big\} 
    = L\trace\bigl( (\THETA \Rzj \THETA^{H} \Rhm)^2 \bigl)\nonumber\\
    &\hspace{6em}+ L^2\bigl( \trace(\THETA \Rzj \THETA^{H} \Rhm) \bigl)^2.             
\end{align}  

Finally, using the independence property between the direct and indirect channels, $\Ex\{ ad \}$ can be derived as
\begin{align}~\label{eq:Ex_2ad}
   &\Ex\{ 2ad \} 
   = L^2 \betamjeu \bigl( \trace( \THETA^H \Rhm \THETA \Rzj) \bigl). 
\end{align}

By substituting~\eqref{eq:Ex_a^2}, \eqref{eq:Ex_b^2}, \eqref{eq:Ex_d^2}, \eqref{eq:Ex_2ad} into \eqref{eq:Ex_gmjue4}, followed by some algebraic transformations, Lemma~\ref{Lemma:2ndand4thmoment} is obtained.

\vspace{-1em}
\section{Proof of Proposition~\ref{Theorem:RF:PPZF}}~\label{appendixHE}
{{First, we compute the first term in \eqref{eq:El_average} as
\begin{align}
    \Ex\big\{ \big\vert \big(\gmjue \big)^{\!H} \wemj^{\PMRT} \big\vert^{2} \big\} 
    &\!=\!
    \Ex\big\{ \big\vert \big(\hgmjue \!+\! \gtilmjeu\big)^{\!H} \wemj^{\PMRT}  \big\vert^{2} \big\} \nonumber\\
    &\!=\! \Ex\big\{\big\vert\big(\hgmjue\big)^{\!H}\wemj^{\PMRT}\big\vert^2\big\}.
\end{align}
According to \eqref{eq:wemrt}, $\Ex\big\{\big\vert\big(\hgmjue\big)^{\!H}\wemj^{\PMRT}\big\vert^2\big\} $ is computed as
\begin{align}
    \Ex\big\{\!\big(\hgmjue\big)^{\!H}\!\wemj^{\PMRT}\!\big\}
    &\!\!\stackrel{(a)}{=} \!\! 
    \frac{
    \Ex\big\{\! \qB_m \!\big\}
    \Ex\big\{\! \hgmjue(\gmjue)^{H} \!\big\}
    }{
    \sqrt{(L\!-K)\gameumj}
    } \!\!=\!\!\sqrt{(L\!-\!K)\gameumj},
    \nonumber
\end{align}
where (a) exploits the independence between $\gmjue$ and the idempotent Hermitian matrix $\qB_{m}$, where $\Ex\big\{\! \qB_m \!\big\}=(L-K/L)\qI_L$ \cite[Lemma 1]{Mohamed_projection_2024}. Moreover, by applying the fourth-order moment rule \cite[Appendix A.2.4]{cite:HienNgo:cf01:2017}, we derive
\begin{align}
    \Ex\big\{\!\big\vert\big(\hgmjue\big)^{\!H}\wemj^{\PMRT}\big\vert^2\!\big\} 
    &= \frac{(L-K)(L-K+1)(\gameumj)^2}{(L-K)\gameumj} \nonumber\\
    &= (L-K+1)\gameumj.
\end{align}

Following the independence property between $\gtilmjeu$ and $\wemj^{\PMRT}$, the second expectation term in \eqref{eq:El_average}, is obtained as
\begin{equation}
    \Ex\big\{ (\gtilmjeu)^{H} \wemj^{\PMRT} (\wemj^{\PMRT})^{H} \gtilmjeu \big\} = \delta_{mj}.
\end{equation}
Similarly, the third expectation term can be computed as
\begin{equation}
    \Ex\big\{ (\gmjue)^{H} \wimk^{\PZF} (\wimk^{\PZF})^{H} \gmjue \big\} = \delta_{mj},
\end{equation}
due to that $\gmjue$ and $\wimk^{\PZF}$ are independent. 
}}

\balance

\end{document}